\documentclass[10pt,letterpaper,twocolumn]{ieeeconf}

\usepackage[utf8]{inputenc} 
\IEEEoverridecommandlockouts               

\usepackage{amssymb,amsmath,amsfonts,mathrsfs}
\usepackage[pdftex,colorlinks]{hyperref}
\usepackage[usenames,dvipsnames]{color}

\newtheorem{theorem}{Theorem}[section]
\newtheorem{remark}{Remark}[section]

\newtheorem{definition}{Definition}[section]

\newcommand{\cH}{\mathcal{H}}

\newcommand{\cF}{\mathcal{F}}

\newcommand{\cS}{\mathcal{S}}
\newcommand{\cL}{\mathcal{L}}

\newcommand{\cK}{\mathcal{K}}

\newcommand{\bbE}{\mathbb{E}}

\newcommand{\bbI}{\mathbb{I}}
\newcommand{\bbP}{\mathbb{P}}

\newcommand{\sF}{\mathscr{F}}
\newcommand{\sG}{\mathscr{G}}

\newcommand{\bra}[1]{\left\langle#1\right|}
\newcommand{\ket}[1]{\left|#1\right\rangle}

\newcommand{\tr}[1]{\mathrm{Tr}\left\{#1 \right\}}
\newcommand{\trc}[2]{\mathrm{Tr}_{#2}\left\{#1 \right\}}
\newcommand{\lsim}{\mathrel{\hbox{\rlap{\lower.55ex \hbox{$\sim$}} \kern-.3em \raise.4ex \hbox{$<$}}}}

\newcommand{\hA}{\hat{A}}

\newcommand{\hB}{\hat{B}}

\newcommand{\hC}{\hat{C}}
\newcommand{\hP}{\hat{P}}
\newcommand{\CP}{\mathrm{CP}}

\newcommand{\hAb}{\hat{\bf A}}
\newcommand{\hBb}{\hat{\bf B}}

\newcommand{\hbbE}{\hat{\bbE}}
\newcommand{\qd}[1]{\langle#1\vert}
\newcommand{\q}[1]{\vert#1\rangle}

\newcommand{\al}[1]{#1}
\newcommand{\zvo}[1]{#1}
\newcommand{\zva}[1]{#1}

\title{\LARGE \bf Quantum filtering using POVM measurements}
\author{Ram A. Somaraju, Alain Sarlette and Hugo Thienpont\thanks{Ram A. Somaraju and Hugo Thienpont are with the Brussels Photonics Team, Dept. of Applied Physics and Photonics, Vrije Universiteit Brussel, Pleinlaan 2-1050 Brussels, Belgium. Alain Sarlette is with the SYSTeMS research group, Faculty of Engineering and Architecture, Ghent University, Technologiepark Zwijnaarde 914, 9052 Zwijnaarde(Ghent), Belgium. {\tt\footnotesize a.somaraju@gmail.com, alain.sarlette@ugent.be, hthienpo@b-phot.org}}\thanks{R. A. Somaraju and H.Thienpont would like to thank Methusalem project IPARC@VUB for financial support as well as the BELSPO IAP project Photonics@BE. A.Sarlette and R. A. Somaraju are members of the BELSPO IAP project DYSCO.}\thanks{This paper presents research results funded by the Interuniversity Attraction Poles Programme, initiated by the Belgian State, Science Policy Office. The scientific responsibility rests with its authors.}}

\begin{document}

\maketitle
\thispagestyle{empty}
\pagestyle{empty}

\begin{abstract} 
The objective of this work is to develop a recursive, discrete time quantum filtering equation for a system that interacts with a probe, on which measurements are performed according to the Positive Operator Valued Measures (POVMs) framework. POVMs are the most general measurements one can make on a quantum system and although in principle they can be reformulated as projective measurements on larger spaces, for which filtering results exist, a direct treatment of POVMs is more natural and can simplify the filter computations for some applications. Hence we formalize the notion of strongly commuting (Davies) instruments which allows one to develop joint measurement statistics for POVM type measurements. \al{This allows us to prove the existence of conditional \zvo{POVMs}, which is essential for the development of a filtering equation. We demonstrate that under generally satisfied assumptions, knowing the observed probe POVM operator is sufficient to uniquely specify the quantum filtering evolution for the system.}
\end{abstract}

\section{Introduction}

The theory of filtering considers the estimation of the system state from noisy and/or partial observations (see e.g.~\cite{Bensoussan1992aa}). For quantum systems, filtering theory was initiated in the 1980s by Belavkin in a series of papers~\cite{Belavkin1980aa,Belavkin1983aa,Belavkin1988aa,Belavkin1992aa}. Belavkin makes use of the operational formalism of Davies~\cite{Davies1976aa}, which is a precursor to the theory of quantum filtering. He has also realized that due to the unavoidable back-action of quantum measurements, the theory of filtering plays a fundamental role in quantum feedback control (see e.g.~\cite{Belavkin1983aa,Belavkin1992aa}). The theory of quantum filtering was independently developed in the physics community, particularly in the context of quantum optics, under the name of quantum trajectory theory~\cite{Braginsky1992aa,Haroche2006aa,Gardiner2004aa,Wiseman2010aa}.

The basic model used to derive filtering equations for a quantum system uses a system-probe interaction. A quantum system, whose state needs to be estimated, is made to interact with a probe and the state of the system becomes entangled with that of the probe. After this interaction, an observable is measured on the probe and this measurement outcome is used to estimate the state of the system. The commutativity of any system observable with any probe observable is used to develop a recursive Markov filtering equation for the system observables (see e.g.~\cite{Bouten2007aa,Bouten2009aa} for an excellent tutorial). 

Suppose $\cH_S$ is the Hilbert space corresponding to the system whose state needs to be estimated and $\cH_P$ is the Hilbert space of the probe. According to the classical von Neumann definition, any probe observable is a self-adjoint operator $Q$ on $\cH_P$; the measurement of such an observable results in an outcome that is (stochastically) an eigenvalue of $Q$, and the probe state after measurement gets projected onto the corresponding eigenspace of $Q$. As far as we are aware, all discussions on quantum filtering theory so far have assumed that the probe undergoes such a von Neumann measurement, also called projective measurement or Projection Valued Measure (PVM). However, a more modern treatment of quantum measurement theory shows that the most general possible quantum measurements that one can perform are the so-called Positive Operator Valued Measures\footnote{The terminology
is not standard and other terms such as Positive Operator Measures or generalized measurements are also used.} (POVMs), of which von Neumann measurements are merely a special case where all the operators are commuting projections~\cite{Davies1976aa,Shapiro1991aa}. See Section~\ref{sec:revPOVM} for a brief overview of POVMs.

POVMs on $\cH_P$ can be reformulated as the restriction to $\cH_P$ of a PVM on a larger space. However, there is no canonical PVM that corresponds uniquely to a given POVM. This is closely related to the fact that the state of a quantum system after a POVM measurement is not uniquely determined as a function of the POVM. To remedy the latter situation, Davies~\cite[Ch.3]{Davies1976aa} has shown that one can associate a (non-unique) instrument to any POVM, which determines a completely positive map that specifies the state after measurement conditioned on the measurement outcome. However, there is again no canonical instrument that corresponds to a given POVM. Therefore, it is impossible to uniquely specify the post-measurement state for a given POVM measurement outcome unless the instrument associated with the measurement is known. As stated by Nielsen and Chuang~\cite[p.~91]{Nielsen2010aa}, ``POVMs are best viewed as [...] providing the simplest means by which one can study general measurement statistics, without the necessity for knowing the post-measurement state.'' As such they are a minimal description of quantum measurements, so one can hope that the POVM formalism leads to more concise and fast filtering equations, suited for (possibly analog) implementation in real-time quantum feedback experiments.

There are other reasons to develop a POVM-based filtering theory that shortcuts the lift to PVMs in larger spaces. For one, once the POVM theory is available it can be more natural to use, as several practical measurement setups are based on POVMs, such as e.g.~approximate position/momentum measurements~\cite[Ch. 3]{Davies1976aa} or phase measurements~\cite{Shapiro1991aa,Berry2002aa}. In some infinite-dimensional situations there can even be conceptual barriers to a PVM viewpoint. Indeed, with phase measurements after much investigation there is still no universal agreement on an acceptable PVM~\cite{Lynch1995aa}. Finally, regarding system identification, the POVM associated with any experimental setup can (at least conceptually) be directly deduced from measurement outcomes; in contrast, in order to ascertain the associated instrument it is necessary to analyze the post-measurement state of the system (for more details see Section~\ref{sub:estPOVM}). This is not generally feasible in practical experimental setups where often the measurements destroy the quantum state (e.g. photo-detection) and/or non-measurably alter it due to interaction with the environment. 

In this paper, we develop a discrete time filtering equation for the system state conditioned on POVM measurements performed on a probe. After reviewing the POVM formalism (Section \ref{sec:revPOVM}), we provide a general theory about (strongly) commuting instruments that are associated to POVMs (Section \ref{sec:commIns}). \al{We then generalize the filtering framework of \cite{Bouten2007aa,Bouten2009aa} from PVMs to POVMs. First Section \ref{sec:condexp} illustrates POVM-specific difficulties in the conditional expectation approach. Section \ref{sec:condprob} then defines conditional probabilities for POVMs.} In the setup consisting of a probe coupled to a target system, any (physically reasonable) instrument associated with a POVM acting only on the probe, strongly commutes with any instrument associated with a POVM acting only on the target system. In Section \ref{sec:Filtering} we show how this allows to define a filtering equation for the system state conditioned on \al{probe} POVM measurement outcomes. \al{This} filtering equation is only a function of the observed probe POVM and does not depend on the associated instrument nor other POVM elements.


\section{Review: POVMs and associated instruments}\label{sec:revPOVM}

The POVM formalism is a standard part of most modern quantum information textbooks. We briefly review it here and refer the interested reader to~\cite[Ch. 3]{Davies1976aa} for more details. 

Consider a quantum system with Hilbert space $\cH$, i.e.~the system state is given by a density operator $\rho$ which is a unit-trace nonnegative self-adjoint linear operator on $\cH$. We use $\phantom{l}^*$ to denote the adjoint. Denote by $\cL(\cH)$ the set of linear operators on $\cH$, 
by $\cL_+(\cH) \subset \cL(\cH)$ the set of self-adjoint nonnegative linear operators, and by $\cS(\cH) \subset \cL_+(\cH)$ the set of all possible density operators (i.e. non-negative trace class operators of unit trace). Standard textbook treatment of quantum measurements assumes that any physically measurable quantity $\hat{A}$ is associated to a self-adjoint operator $A:\cH\to \cH$. Because $A$ is self-adjoint, we have the spectral decomposition\footnote{For clarity of explanation, we assume that $A$ has a discrete spectrum. The discussion easily generalizes to the continuous spectrum situation.} 
\begin{equation}\label{eqn:spec}
A = \sum_{\omega\in \Omega} \omega P_{\omega}
\end{equation}
where $\Omega$ is the set of eigenvalues of $A$ and $P_\omega$ is the eigenprojection corresponding to eigenvalue $\omega$. Starting with a system in state $\rho$, according to von Neumann's measurement postulates we have: 
\begin{enumerate}
\item any measurement of the observable $A$ gives some outcome $\omega\in \Omega$ with probability $\tr{\rho P_\omega}$, and 
\item after measurement outcome $\omega$, the state of the system becomes
$$\rho' = \frac{P_\omega\rho P_\omega}{\tr{P_\omega\rho P_\omega}}\, .$$
\end{enumerate}

The first postulate can be thought of as:
$\Omega$ is the set of all possible measurement outcomes of an experimental setup $(\hat{A})$ and to each $\omega\in\Omega$, one assigns a projection $P_\omega$ in 
$\cH$ such that $\tr{\rho P_\omega}$ is the probability of measuring $\omega$. This motivates the following generalization.
\begin{definition}\cite[Def 3.1.1]{Davies1976aa} Let $\Omega$ be a set, $\sF$ a $\sigma$-field of subsets of $\Omega$, and $\cH$ a Hilbert space. Then a $\cH$-valued Positive Operator Valued Measure (POVM) on $\Omega$ is a map $\hA:\sF\to \cL_+(\cH) ;\ E\to \hA(E) = \hA_E$ such that 
\begin{enumerate}
\item $\hA(E) \geq \hA(\emptyset) = 0$ for all $E\in \sF$;
\item For any countable, mutually disjoint collection $\{E_n\} \subset \sF$ we have
$\; \hA\left(\bigcup_{n} E_n\right) = \sum_n \; \hA(E_n)$ \newline
where the series convergence on the right is in the weak operator topology;
\item $\hA(\Omega) = \bbI_\cH$, the identity operator on $\cH$. 
\end{enumerate}
\end{definition}
A POVM that corresponds to a physical experiment has a simple interpretation. The set $\Omega$ is the sample space corresponding to experimental outcomes so that the $\sigma$-field $\sF$ consists of the set of all events. The POVM $\hA$ and a state $\rho$ on $\cH$ induce a measure $\mu_{\rho,\hA}(\cdot) = \tr{\rho\hA(\cdot)}$ on $\Omega$ so that $\mu_{\rho,\hA}(E)$ gives the probability of event $E\in \sF$. We use the notation $\hA\in E$ to denote the event that the measurement of POVM $\hA$ resulted in a value in $E\in \sF$.

Note that here $\Omega$ can have any general structure. This allows one to describe measurement apparatuses with outcomes that are physically e.g.~multi-dimensional or on a manifold topology like the circle or sphere, which is not possible with standard von Neumann measurements. The latter are indeed equivalent to a special case of POVMs called \emph{Projection Valued Measures (PVMs)}, which require that $\Omega$ is a closed subset of the real line and that the range of $\hA$ only consists of commuting projections. The unique correspondence between PVMs ($\hat{A}$) and self-adjoint operators describing von Neumann measurements ($A$) is obtained through \eqref{eqn:spec} by setting $P_\omega = \hA(\omega)$ for each $\omega \in \Omega$.

It has been shown that any POVM on $\cH$ can be viewed as the restriction to $\cH$ of a PVM on a larger Hilbert space.
\begin{theorem}\cite[Th.9.3.2]{Davies1976aa}\label{the:POVMIns} Let $\Omega$ be a compact metrizable space with Borel field $\sF$, and $\hA$ a POVM taking values in $\cL_+(\cH)$. Then there exists a Hilbert space $\cK\supset \cH$ and a PVM $\hA^{\cK}: \sF\to \cL_+(\cK)$ such that if $P$ is the orthogonal projection from $\cK$ onto $\cH$ then $\hA(E)$ is the restriction of $P\hA^{\cK}(E)P$ to $\cH$ for all Borel sets $E$.
\end{theorem}
In principle, the existing filtering theory for PVMs \cite{Bouten2007aa,Bouten2009aa} thus covers the needs of POVM-based filtering, modulo a proper lift of the Hilbert space. However, the latter is not unique and for reasons explained in Section I, it makes sense to look for a POVM theory that does not build on its reduction from a PVM.\vspace{3mm}

To generalize the second measurement postulate, Davies~\cite{Davies1976aa} introduces the notion of an instrument as a complement to POVMs. In the following we denote by $\CP(\cH)$ the set of all completely positive (CP) maps\footnote{CP maps are the most general quantum evolutions, see~\cite[Sec.9.2]{Davies1976aa} and~\cite[Ch.8]{Nielsen2010aa} for a discussion. The definition given below for instruments is different from that given in~\cite{Davies1976aa}, wherein the assumption of complete positivity is replaced by positivity; see the discussion in~\cite[Sec.9.2]{Davies1976aa}.} ${\bf A}:\rho\mapsto {\bf A}(\rho)$ on states $\rho\in \cS(\cH)$.
\begin{definition}\cite[Def 4.1.1]{Davies1976aa} Let $\Omega$ be a set, $\sF$ a $\sigma$-field of subsets of $\Omega$, and $\cH$ a Hilbert space. Then a $\cH$-valued instrument on $\Omega$ is a map $\hAb:\sF\to \CP(\cH);\ E \to \hAb_E(\cdot)$ such that 
\begin{enumerate}
\item $\hAb_E \geq \hAb_\emptyset = 0$ for all $E\in \sF$;
\item For any countable, mutually disjoint collection $\{E_n\} \subset \cF$ we have
$\; \hAb_{(\cup_{n} E_n)} = \sum_n \; \hAb_{E_n}$\newline
where the series convergence on the right is in the weak operator topology;
\item $\tr{\hAb_\Omega(\rho)} = \tr{\rho}$, for all $\rho\in \cS(\cH)$. 
\end{enumerate}
\end{definition}
If an experiment is set up so that the outcomes take values in some set $\Omega$, then for a quantum system prepared in state $\rho\in \cS(\cH)$ the measurement postulates for an instrument write:
\begin{enumerate}
  \item an outcome in the set $E\subset \Omega$ is obtained with probability $\bbP(E) = \tr{\hAb_E(\rho)}$;
  \item the state of the system conditioned on a measurement outcome in set $E$ is $\hAb_E(\rho)/\bbP(E)$.
\end{enumerate}

\begin{theorem}\cite[Th.3.1.3, Th.9.2.3]{Davies1976aa}\label{the:InsPOVM} 
If $\hAb$ is an instrument then for all $E\in \sF$, there exits a (non-unique) countable set $\{ A_{n}(E) \}_{n\in N}\subset \cL(\cH)$ such that
\begin{equation}\label{newrep}
\hAb_E(\rho) = \sum_{n\in N} A_{n}(E) \rho A_{n}(E)^*\, .
\end{equation}
Moreover, there exists a unique POVM $\hA$ on $\Omega$ associated to $\hAb$ such that for all $E\subset \Omega$ and $\rho \in \cS(\cH)$ we have:
\begin{equation}\label{eq:IPcorr}
\tr{\hAb_E(\rho)} = \tr{\rho \hA_E} \, .
\end{equation}
In particular, this unique POVM is given by
$$\hA(E) = \sum_{n\in N} A_n(E)^* \,A_n(E) \textrm{ for all } E\in \sF.$$
\end{theorem}
Theorem~\ref{the:POVMIns} implies that one can also always construct an instrument corresponding to a given POVM. However, there is no unique nor canonical way to choose the instrument without further information about the physical system.


\subsection{System identification of POVMs}\label{sub:estPOVM}
	
Consider an experimental setup corresponding to unknown instrument $\hAb$, associated to POVM $\hA$. In order to experimentally determine $\hA$, one can initialize the system being measured in some state $\rho = \ket{\phi}\bra{\phi}$ and for any $E\in \sF$, the probability of measurement outcome in set $E$ is
$$\bbP[\hA\in E] = \bra{\phi} \hA(E) \ket{\phi}.$$ 
With sufficiently many experimental outcomes, one can estimate the classical probability distribution $\bra{\phi} \hA(E) \ket{\phi}$ over all $E \subset \Omega$; doing this for different $\ket{\phi}$ and using polarization then allows to calculate $\hA(E)$ itself. In order to ascertain the instrument $\hAb$ however, we must have access to the state $\hAb_E(\rho)$ e.g.~performing a state tomography experiment. This is often impractical in experimental setups.
\zvo{
\begin{remark}\label{rem:povmIns}
In the following sections, we use instruments in theoretical developments and to examine general properties of the measurement settings we consider. The goal is however to show that our final filtering equation uses POVM data only.
\end{remark}}


\subsection{Notation} 
In the remainder of the paper we will use $\hAb,\hBb,\ldots$ to denote instruments and $\hA,\hB,\ldots$ to denote the POVMs corresponding to these instruments. Also, if an instrument $\hAb$ corresponds to a von Neumann measurement, then we denote by $A$ the associated self-adjoint operator. We will use the terms PVM and self-adjoint operator interchangeably.


\section{Commuting Instruments}\label{sec:commIns}

\al{The classical development of filtering equations builds} on joint probabilities, which are not obvious in the quantum context. Therefore the central idea in~\cite{Bouten2007aa} is that in order to define a conditional expectation of two self-adjoint operators, the two operators must commute with each other. \zvo{The von Neumann measurement postulates imply that `joint' measurement statisics can be defined only in this situation.}

We now wish to generalize the \al{filtering framework} of \cite{Bouten2007aa,Bouten2009aa} towards POVMs and for this we need to understand when it is possible to measure two POVMs simultaneously. In fact this depends on the commutativity of the instruments used to implement the POVMs.
\begin{definition}\label{def:comm}
Suppose $(\Omega_1,\sF_1)$ and $(\Omega_2,\sF_2)$ are two measure spaces and $\cH$ is some Hilbert space. Then two $\cH$-valued instruments, $\hAb_1:\sF_1\to \CP(\cH)$ and $\hAb_2:\sF_2\to \CP(\cH)$ are said to  \emph{strongly commute} if for all $E_1\in \sF_1$ and $E_2\in \sF_2$, there exist sequences of operators $\{A^1_m(E_1):m\in N_1\}$ and $\{A^2_m(E_2):m\in N_2\}$ in $\cL(\cH)$ such that:
\newline $\bullet$ the instruments write (cf. Theorem~\ref{the:InsPOVM})
\begin{eqnarray}
\label{asn1} \hAb_{i,(E_i)}(\rho) = \sum_{m\in N_i} {A^i_m}(E_i)\, \rho \, A^i_m(E_i)^* \quad \forall \rho\, , \; i=1,2
\end{eqnarray}
$\bullet$ for all $m\in N_1$ and $n\in N_2$ the commutator $[,]$ gives:
\begin{eqnarray}\label{asn3}
[A^1_m(E_1),A^2_n(E_2)] & = & [A^1_m(E_1),{A^2_n}(E_2)^\ast] = 0 \, .
\end{eqnarray}
\end{definition}
\vspace{3mm}

\begin{remark}
For the special case of PVMs, $N_i=1$ for all $i$ and the commutativity condition of Def.~\ref{def:comm} is clearly equivalent to the commutativity of the associated self-adjoint operators $A_i$. \zvo{Note that in general, an instrument does not strongly commute with itself (cf.~proof of Theorem~4.3.1 in~\cite{Davies1976aa})}.
\end{remark}

Now we consider the composition of instruments --- the filtering application will involve one (actual) instrument on the probe and one (hypothetical, expressing our goal-variable) on the target system. 
\begin{theorem}\cite[Th.3.4.2]{Davies1976aa}\label{the:insComp}
Suppose $\hAb_i$, $i=1,2,\ldots n$ are instruments on some compact metrizable $\Omega_i$ with Borel field $\sF_i$. Then there exists a unique ``joint'' instrument $\hAb$ on $\Omega_1\times \Omega_2\times ... \times \Omega_n$ such that for all $E_i\in \sF_i$ and $\rho$ we have:
$$\hAb_{E_1\times E_2\times\cdots\times E_n} (\rho) = \hAb_{n,(E_n)} \circ ... \circ \hAb_{2,(E_2)} \circ \hAb_{1,(E_1)} (\rho)\, .$$ 
\end{theorem}\vspace{2mm}
We now prove the first result of this paper.
\begin{theorem}\label{the:jointDist}
Suppose $\hAb_i$, $i=1,2,\ldots p$ are instruments on some compact metrizable $\Omega_i$ with Borel field $\sF_i$ and $\hA_i$ are the corresponding POVMs. If the $\hAb_i$ are pairwise strongly commutative, then the POVM $\hA$ corresponding to the joint instrument $\hAb$ is uniquely determined by the POVMs $\hA_i\,$, according to 
$$\hA(E_1\times E_2\times...\times E_p) = \hA_1(E_1)\hA_2(E_2)\, ... \,\hA_p(E_p) \, .$$
Moreover, the POVMs are mutually commutative, that is $[ \hA_i(E_i),\hA_j(E_j)] = 0$ for all $E_i,\; E_j,\; i\neq j$.
\end{theorem}
\begin{proof}
We first prove the result for $p=2$. 
Select some events $E_1\in \sF_1$, $E_2\in \sF_2$ and let $E=E_1 \times E_2$.
From Definition~\ref{def:comm} construct sequences of operators $\{A^1_n:n\in N_1\}$ and $\{A^2_n:n\in N_2\}$ satisfying \eqref{asn1},\eqref{asn3}, where for notational convenience we have written $A^i_n$ for $A^i_n(E_i)$, with $i=1,2$.
Then for all $\rho\in \cS(\cH)$, we have
\begin{eqnarray}
\nonumber \tr{\rho \hA_E} &=& \tr{\hAb_E(\rho)}\\
\nonumber \text{\scriptsize{(Th.\ref{the:insComp})}} &=& \tr{\hAb_{2,(E_2)} \circ \hAb_{1,(E_1)} (\rho)}\\
\nonumber \text{\scriptsize{(Th.\ref{the:InsPOVM})}} &=& \tr{ \sum_{m\in N_1,n\in N_2} A^2_n A^1_m\, \rho\, {A^1_m}^* {A^2_n}^*}\\
\nonumber \text{\scriptsize{(trace property)}} &=& \tr{ \rho \sum_{m\in N_1,n\in N_2} {A^1_m}^* {A^2_n}^* \, A^2_n A^1_m}\\
\nonumber \text{\scriptsize{(Def.\ref{def:comm}, \eqref{asn3})}} &=& \tr{ \rho \sum_{m\in N_1} {A^1_m}^*A^1_m \sum_{n\in N_2} {A^2_n}^* A^2_n }\\
\label{thatslast} &=& \tr{\rho \hA_1 \hA_2} \, = \,\tr{\rho \hA_2 \hA_1} \, .
\end{eqnarray}
To see how the recursive argument works, consider $p=3$. The above leads to
$\tr{\rho \hA_E} = \tr{\hAb_{1,(E1)}(\rho) \hA_2 \hA_3} \, .$ 
Choosing representations \eqref{newrep} such that instruments $\hAb_1$ and $\hAb_2$ commute, we get
$\tr{\rho \hA_E} = \tr{\hAb_{1,(E1)}(\rho \hA_2) \hA_3} \, .$
Now choosing representations \eqref{newrep} such that $\hAb_1$ and $\hAb_3$ commute, we get the result.
\end{proof}


\section{Conditioning with respect to a POVM}

We are now in a position to define the conditional POVM associated with two strongly commuting instruments. Recall that we wish to find an expression for the conditional POVM that is expressed only in terms of the POVMs and not in terms of the instruments themselves (cf. Remark~\ref{rem:povmIns}). 


\subsection{Basic definitions}\label{sec:basDef}

In this section, let $\cH$ be a Hilbert space, $\Omega_A$ and $\Omega_B$ be two compact metrizable sets with Borel algebras $\sF_A$ and $\sF_B$, respectively, and $\hA$ and $\hB$ be two $\cH$-valued POVMs on $\Omega_A$ and $\Omega_B$ corresponding to the instruments $\hAb$ and $\hBb$.

We denote $\rho$ a state on $\cH$ and introduce the semi-norm $\|\cdot\|_\rho = |\tr{\cdot \rho} |$ on the set of bounded operators on $\cH$. Two operators $\Gamma_1,\Gamma_2$ on $\cH$ are said to be \al{$\rho$-equivalent} (written $\Gamma_1 \equiv \Gamma_2$) if $\|\Gamma_1-\Gamma_2\|_\rho = 0$. If $(\Omega,\sF,\mu)$ is a probability space then two functions $f,g:\Omega\to \cL^+$ are said to be $\rho$-equivalent if $\|f(x)-g(x)\|_\rho= 0$ for $\mu$-a.e. $x\in \Omega$. 

If $f$ is any measurable function on $\Omega_A$, then the integral of $f$ with respect to $\hA$ over a set $E\in \sF_A$ is defined by:\footnote{It should also be possible to define this integral as a Stiltjes integral over the measurable space $\cL_+(\cH)$ with the measure induced by the POVM.}
$$\int_{\omega\in E}f(\omega) \tr{\rho\,\hA(d\omega)} \triangleq \int_{\omega\in E} f(\omega) d\mu_{\rho,\hA}(\omega)$$
where the measure $\mu_{\rho,\hA}(\cdot) \triangleq \tr{\rho \hA(\cdot)}$ is a probability measure on $(\Omega_A,\sF_A)$. 

\subsection{Conditional Expectation}\label{sec:condexp}

Following~\cite{Bouten2009aa}, one approach to conditioning and filtering is through the conditional expectation. When a vector space is associated to the POVM measurement results $\Omega_A$, we can compute the \emph{expectation value of a POVM $\hA$} in state $\rho$ by:
$$\bbE_\rho[\hA] = \int_{\omega\in \Omega}\omega \, \tr{\rho\, \hA(d\omega)} \, .$$

If $\hAb$ and $\hBb$ are strongly commuting instruments, we can set $\Omega=\Omega_A\times \Omega_B$ with the product $\sigma$-field and define $\hA\hB$ the product POVM as in Theorem~\ref{the:jointDist}. On the product space $\Omega$ we can then define two classical random variables $\alpha:(\omega_A,\omega_B) \mapsto \omega_A$ and $\beta:(\omega_A,\omega_B)\mapsto \omega_B$, and we know from the classical Kolgomorov theory that the conditional expectation
$\bbE [\alpha|\beta]: \Omega \to \Omega_A$ exists and is a (a.e.) unique random variable that is measurable with respect to $\sF_{\beta} \equiv \sF_B$. Defining $\sG_A$ the $\sigma$-algebra of subsets of $\Omega_A$ generated by $\bbE[\alpha|\beta]$,\footnote{Thus, ``$\sG_A$ is as coarse-grained as $\sF_B$ or coarser'', in the sense that it is the $\sigma$-algebra generated by an $\sF_B$-measurable function.} we can lift this to a well-defined \emph{conditional expectation POVM:}
\begin{eqnarray}
	\label{eqn:condPOVMdef} \hbbE_{\hA \vert \hB}(E) =  \hB\left(\bbE[\alpha\big|\beta]^{-1}(E)\right) \quad \text{ for any } E \in \sG_A\, .
\end{eqnarray}
This basically attributes to event $E$ the $\hB$-POVM element associated to the union of $\beta$ for which $\bbE[\alpha|\beta] \in E$.

\al{However, unlike in the PVM case~\cite{Bouten2009aa},  $\hbbE_{\hA \vert \hB}$ is not trivial to use for filtering purposes. Indeed, because in general $\hB(E_1) \hB(E_2) \neq 0$ even for disjoint events $E_1,E_2$, it is not clear how to input an actual measurement result for $\hB$ into this expression.}

\subsection{Conditional \zva{POVMs}}\label{sec:condprob}


\zva{We now consider an approach that is motivated from classical conditional probability.}
\begin{theorem}\label{thm:conpro}
 Let $\hAb$, $\hBb$ as defined in Section~\ref{sec:basDef} be strongly commuting instruments. There exists a (not necessarily unique) map
$\hP:\sF_A\times \Omega_B\to \cL^+$ 
such that: 
\begin{enumerate}
\item For all $E\in \sF_A$ and $F\in \sF_B$, we have 
$$\hA(E)\hB(F) \equiv \int_F \hP(E,\omega) \tr{\rho \hB(d\omega)}.$$
\item For $\mu_{\rho,\hB}$-almost-every $\omega\in \Omega_B$ and every mutually disjoint sequence $E_1,E_2,\ldots\in \sF_A$ we have the countable additivity condition
$$\hP\left(\bigcup_{n=1}^\infty E_n, \omega\right) \equiv \sum_{n=1}^\infty \hP(E_n,\omega).$$
\item $\hP(\Omega_A,\omega) \equiv \bbI_\cH$ for $\mu_{\rho,\hB}$-a.e. $\omega\in \Omega_B$.
\end{enumerate}
If $\hP'$ is another such map then $\hP$ and $\hP'$ are $\rho$-equivalent.
\end{theorem}
\begin{definition}\label{def:conpro}
We call \al{$\hP_{A \vert B} \triangleq \hP$} the \emph{conditional POVM of $\hA$ given $\hB$}. \zva{Then $\tr{\rho\; \hP_{A \vert B}(E,\omega)} \!= \mathbb{P}(\hA \!\in\!\! E \vert \hB \!\in\!\! \{\omega\})$ gives}
the probability that a measurement of $\hA$ is in $E$ \al{when knowing that $\omega$ is obtained from a $\hB$ measurement.}
\end{definition}

\begin{proof}
The proof follows~\cite{Kupka1972aa} and uses the notion of a lifting. Let $(\Omega,\sF,\mu)$ be a measure space and denote $E\thicksim_\mu F$ if $E \in \sF$ and $F \in \sF$ differ by a $\mu$-measure zero set. A lifting of $\mu$ is a map $\phi:\sF\to \sF$ such that 
\begin{enumerate}
\item $\phi(E) \thicksim_\mu E$.
\item $E\thicksim_\mu F$ implies $\phi(E) = \phi(F)$.
\item $\phi(E\cup F) = \phi(E)\cup \phi(F)$ and $\phi(E\cap F) = \phi(E)\cap \phi(F)$.
\item $\phi(\emptyset) = \emptyset$ and $\phi(\Omega) = \Omega$.
\end{enumerate}

Denote $\hC(E,F) = \hA(E)\hB(F)$ the joint POVM on $\Omega_A \times \Omega_B$, see Theorem~\ref{the:jointDist}, and let $\mu_{\rho,\hA}$, $\mu_{\rho,\hB}$ and $\mu_{\rho,\hC}$ be the associated probability measures on $\Omega_A$, $\Omega_B$ and $\Omega_A\times \Omega_B$. 

Because $\Omega_A$ is a complete metric space we can apply Theorem 6.6.6 in~\cite{Ash1972aa} and this implies there exists a regular conditional probability $\bbP:\sF_A\times \Omega_B\to [0,1]$ which is $\sF_B$-measurable for any $E\in\sF_A$ and satisfies 
\begin{enumerate}
\item $\mu_{\rho,\hC}(E,F) = \int_F\bbP(E,\omega) d\mu_{\rho,\hB} (\omega)$ for all $E\in \sF_A$ and $F\in \sF_B$.
\item For every mutually disjoint sequence $E_1,E_2,\ldots\in \sF_A$ and $\mu_{\rho,\hB}$-a.e. $\omega$ we have 
$$\bbP\left(\bigcup_{n=1}^\infty E_n, \omega\right) \equiv \sum_{n=1}^\infty \bbP(E_n,\omega).$$
\item For $\mu_{\rho,\hB}$-a.e. $\omega$ we have $\bbP(\Omega_A,\omega) = 1$.
\end{enumerate}

Following~\cite{Kupka1972aa}, fix $E\in \sF_A$, denote $S$ the support of $\bbP(E,\cdot)$ in $\Omega_B$ and let $\phi$ be a lifting for $\mu$ restricted to the measurable subsets of $S$. Consider $\Pi$ be the collection of partitions of $S$ such that $\pi= \{S_1,\ldots,S_n\}\in \Pi$ implies $\phi(S_i) = S_i\neq \emptyset$ for all $i$. Such $\Pi$ always exists and is a directed set under refinement~\cite{Kupka1972aa}. Define
\begin{eqnarray}
\nonumber \bbP_\pi(E,\cdot) &=& \sum_i \frac{\mu_{\rho,\hC}(E,S_i)}{\mu_{\rho,\hB}(S_i)}\chi_{S_i}(\cdot) \;\; \text{ and}\\
\label{eq:tolim} \hP_\pi(E,\cdot) &=& \sum_i \frac{\hC(E,S_i)}{\mu_{\rho,\hB}(S_i)}\chi_{S_i}(\cdot)
\end{eqnarray}
where $\chi_{S_i}$ is the characterisitc function of the set $S_i$. From the definition of $\hC$, we have $\mathrm{Tr}\{\rho\hP_\pi \} = \bbP_\pi$ and by \cite[Lemma 4.3]{Kupka1972aa} the limit exists and gives
\begin{equation}\label{eq:KupkaLimit}
\lim_\pi \tr{\rho \hP_\pi(E,\cdot)} = \lim_\pi \bbP_\pi(E,\cdot) = \bbP(E,\cdot) \, .
\end{equation}
The set $\cL^+_1$ of positive operators on $\mathcal{H}$ with norm less than 1 is a closed and bounded subset \zva{of $\mathcal{L}(\mathcal{H})$.} 
Therefore, it is compact in the weak-$\ast$ topology\footnote{A sequence $A_n\in \cL^+$ converges to $A\in \cL$ in the weak-$\ast$ topology if and only if $\tr{\rho A_n}$ converges to $\tr{\rho A}$ for all trace class operators $\rho$. We write this convergence $\lim^\ast_n A_n = A$.} by the Banach-Alaoglu theorem (see e.g.~\cite[Sec 3.15]{Rudin1973aa}), so there exists a (not necessarily unique) positive operator $\hP(E,\cdot)$ such that 
$$\lim^\ast_\pi \hP_\pi(E,\cdot) = \hP(E,\cdot).$$
By definition of weak-$\ast$ convergence we have $\lim_\pi \mathrm{Tr}\{ \rho\hP_\pi(E,\cdot) \} = \mathrm{Tr}\{ \rho \hP(E,\cdot)\}$ and by \eqref{eq:KupkaLimit} we get $\mathrm{Tr}\{\rho \hP(E,\cdot)\} = \bbP(E,\cdot)$. The proof now follows from the properties of $\bbP$. 
\end{proof}

\zvo{
\begin{definition}
The conditional expectation of $\hA$ \al{knowing that $\hB\in \{\omega_B\}$} is given by 
$$\bbE_\rho[\hA|\hB\in\{\omega_B\}] = \int_{\omega_A\in\Omega_A} \omega_A \tr{\rho\, \hP_{A\vert B}(d\omega_A,\omega_B)}.$$
\end{definition}}


\section{Filtering with POVMs}\label{sec:Filtering}

Our measurement model is motivated from the discrete-time model used in~\cite{Bouten2009aa} for filtering using PVMs. We consider a system with Hilbert space $\cH_S$ and a probe consisting of a sequence of subsystems $n=1,2,...$, each with Hilbert space $\cH_n = \cH$. So the probe is described on a Hilbert space $\cH_P = \otimes_{n=1}^\infty\cH_n$; the combined state space of the probe and system is written $\cH_{tot} = \cH_S \otimes \cH_P$. \al{In the practical setup, we will assume that system and probe are initially separated, and at consecutive times $n=1,2,...$ the system undergoes an interaction with probe subsystem $\cH_n$, which is then measured by a POVM. We want to know the evolution of $\rho^S_n$, the system state after $n$ interactions and probe measurements, conditioned on the latters' outcomes.}

We therefore suppose  that $\Omega_n$ is a compact metrizable space\footnote{If $\Omega_n$ is not compact then we can simply consider the 1-point compactification of $\Omega_n$~\cite[p.12]{Davies1976aa}.} for $n=1,2,...$, let $\sF_n$ a $\sigma$-field on $\Omega_n$ and $\hBb_n:\sF_n \to \CP(\cH_n)$ an instrument with corresponding POVM $\hB_n$. Also, let $(\Omega_A,\sF_A)$ be some measure space and $\hAb:\sF_A\to\CP(\cH_S)$ any system instrument with corresponding POVM $\hA$. Note that we do not associate the $\Omega_n$ to a vector space; this allows to consider probe measurements with results on manifolds, like the circle for phase measurements.

To consider the `active' part while the sequence of interactions progresses, we set $\cH_{n]} = \cH_S$ for $n=0$ and recursively define $\cH_{n]} = \cH_{n-1]}\otimes \cH_n$ for $n\geq 1$; similarly define $\Omega_{n]} = \Omega_A\times\Omega_1\times\cdots\times\Omega_n$. 
Consider the initial state on $\cH_{tot}$,
$$\rho^{tot}_0 = \rho^S_0 \,{\textstyle \bigotimes_{n=1}^\infty}\, \rho^P_n.$$
We suppose that between time steps $n$ and $n+1$, the system interacts with the probe according to a unitary evolution operator $U_n$ on $\cH\otimes \cH_n$, i.e.~it interacts only with subsystem $n$ of the probe. After this unitary evolution, the POVM $\hB_n$ is measured to be some $\omega_n \in \Omega_n$. With a slight abuse of notation, let $U_{n]} = \prod_{i=1}^n U_n$, a unitary on $\cH_{n]}$.

\al{It seems physically reasonable to assume that the instrument associated in this setup to the $\cH_{\text{tot}}$-POVM $\hB_n$ acts non-trivially only on $\cH_n$, i.e.~its operator-sum representation has elements $\{ B^n_i \}_i$ of the form
$$B^n_i = \bbI_{\cH_S} \,({\textstyle \bigotimes_{m=1}^{n-1}} \, \bbI_{\cH}) \;\otimes  b^n_i \, ({\textstyle \bigotimes_{m=n+1}^\infty}  \, \bbI_{\cH}) \; .$$ 
One formal argument for this is that the measurement generally acts at a place away from the system, and often also away from the other (e.g.~travelling or already destroyed) probe subsystems. Similarly, when speaking of system properties through the $\cH_{\text{tot}}$-POVM $\hA$, it makes sense to assume that the associated instrument $\hAb$ acts nontrivially only on $\cH_S$.}  Therefore the set of instruments $\hAb$, $\hBb_n$ for $n=1,2,...$ would strongly commute. Then the same commutations apply to the evolved instruments \zva{where the effect of $n$ interactions is described in the Heisenberg picture,}
$$\hAb(n) \triangleq U_{n]} \hAb  U_{n]}^* \quad \text{and} \quad \hBb_i(n) \triangleq U_{n]} \hBb_i  U_{n]}^* \, .$$

\al{Therefore, the conditional probability of $\hA(n)$ with respect to $\hB_1(n)\hB_2(n)\cdots \hB_n(n)$ is well-defined according to Theorems \ref{the:jointDist} and \ref{thm:conpro}: there exists a function 
$$\hP^A_{n}\left(E,\omega_{n]}\right) \triangleq \hP_{\hA(n)\vert \hB_1(n)\hB_2(n)\cdots \hB_n(n)}\left(E,\omega_{n]}\right)$$
such that $\tr{\hP^A_{n}\left(E,\omega_{n]}\right) \rho_0^{tot}}$ gives the probability of $\hA$-events, knowing the outcomes $\omega_{n]} \triangleq\omega_1,...,\omega_n$ of the $n$ first probe measurements after interactions with the system. 
From \eqref{eq:tolim}, the map $\Phi_{\rho_0^{tot}} : \hA \to \bbE_{\rho_0^{tot}} [\hA|\omega_{n]}]$ is linear and real-valued.\footnote{Note that $\Phi$ is unique since all possible $\hP$ are $\rho_0^{tot}$-equivalent.}  Thus by Gleason's theorem, at least for finite-dimensional $\cH_S$, there exists a unique density operator, which we identify as \emph{the post-measurement state $\rho^S_n(\omega_{n]})$}, such that $\Phi_{\rho_0^{tot}}(\hA) = \tr{\hA \, \rho^S_n(\omega_{n]})}$ for all $\hA$. Algebraic computations based on property 1) of Theorem \ref{thm:conpro} then lead to the well-known expression:
$$\rho^S_n(\omega_{n]}) \, = \,  \frac{\trc{U_{n]} \rho_0^{tot} U_{n]}^\ast \; \hB(\omega_{n]})}{P}}{\tr{U_{n]} \rho_0^{tot} U_{n]}^\ast \; \hB(\omega_{n]})}}\, ,$$
\zva{where $\trc{\cdot}{P}$ is the partial trace over $\cH_P$.}
Note that, thanks to the strong commutativity condition, this expression depends on the POVM but not on the associated instrument. Moreover, in this expression only the POVM element associated to the actually observed $\omega_{n]}$ is needed, irrespective of the other potentialities completing the POVM. Interesting implications of this are best illustrated by the following example.

\zva{\emph{Qubit phase:} $\;$ Consider a probe composed of qubits, $\cH = \text{span}_{\mathbb{C}}\{\q{0},\q{1}\} \cong\mathbb{C}^2$, on which we apply the POVM with $M$ elements $d=0,...,M-1$,
$$\hB_{\cH}(d) =  \tfrac{1}{M} (\q{0}+e^{2i\pi d/M}\q{1})(\qd{0}+e^{-2i\pi d/M}\qd{1}) .$$
Then the effect on the system of a detection result e.g.~$\omega_1=0$ is the same, whether $M=2$ (projective measurement), or $M=3,4,...$ or any larger number. It seems legitimate to attribute the same effect to the continuous POVM limit.}}

\section{Conclusion and Future work}\label{sec:concl}

In this paper we show how quantum filtering can be performed in the POVM setting.
We formalize a notion of strongly commuting instruments, which gives a sufficient condition to define joint measurement statistics in terms of the associated POVMs only, without explicitly depending on the instruments. \al{We introduce the notion of conditional-expectation-POVM for measurements with commuting instruments, and highlight that it can be inappropriate for filtering, unlike in the PVM case. We then show that the notion of conditional probabilities can be defined for strongly commuting instruments. On that basis, for a system-probe model, we analyze filtering of the system state conditioned on POVM probe measurements, and highlight its general properties. A future goal is to apply these ideas to derive a filtering equation for discrete-time \emph{phase} measurements, a quintessential example of a POVM-type measurement where the associated instrument is not known.} On a more theoretical mode, we also should explore \emph{necessary} conditions for the \zva{strong} commutation of two instruments.

 \bibliographystyle{IEEEtran}

\end{document}